\documentclass[conference]{IEEEtran}

\IEEEoverridecommandlockouts

\usepackage{amsmath,amsthm,relsize}
\usepackage{amsfonts, amssymb, cuted}
\usepackage{cleveref}
\usepackage{mathtools}
\usepackage{algorithm}
\usepackage[noend]{algpseudocode}
\usepackage{cite}
\usepackage{xcolor}
\usepackage{soul}
\usepackage{graphicx}
\usepackage{tikz}
\usepackage{pgfplotstable}
\usepackage{pgfplots}
\usepackage{nicefrac}
\usepackage{enumitem}
\usepackage{support-caption}
\usepackage{subcaption}
\usetikzlibrary{plotmarks}
  \usetikzlibrary{arrows.meta}
  \usepgfplotslibrary{patchplots}
  \usepackage{grffile}
  \pgfplotsset{plot coordinates/math parser=false}
  \newlength\figureheight
  \newlength\figurewidth
\usepackage[font=footnotesize]{caption}
\usepackage{multirow}
\usepackage{acro}
\usepackage{pifont}

\newtheorem{lemma}{Lemma}

\newcommand{\vbar}{\raisebox{.17ex}{\rule{.04em}{1.35ex}}}
\newcommand{\vbarind}{\raisebox{.01ex}{\rule{.04em}{1.1ex}}}
\newcommand{\R}{\ifmmode{\rm I}\hspace{-.2em}{\rm R} \else ${\rm I}\hspace{-.2em}{\rm R}$ \fi}
\newcommand{\T}{\ifmmode{\rm I}\hspace{-.2em}{\rm T} \else ${\rm I}\hspace{-.2em}{\rm T}$ \fi}
\newcommand{\N}{\ifmmode{\rm I}\hspace{-.2em}{\rm N} \else \mbox{${\rm I}\hspace{-.2em}{\rm N}$} \fi}
\newcommand{\B}{\ifmmode{\rm I}\hspace{-.2em}{\rm B} \else \mbox{${\rm I}\hspace{-.2em}{\rm B}$} \fi}
\newcommand{\Hil}{\ifmmode{\rm I}\hspace{-.2em}{\rm H} \else \mbox{${\rm I}\hspace{-.2em}{\rm H}$} \fi}
\newcommand{\C}{\ifmmode\hspace{.2em}\vbar\hspace{-.31em}{\rm C} \else \mbox{$\hspace{.2em}\vbar\hspace{-.31em}{\rm C}$} \fi}
\newcommand{\Cind}{\ifmmode\hspace{.2em}\vbarind\hspace{-.25em}{\rm C} \else \mbox{$\hspace{.2em}\vbarind\hspace{-.25em}{\rm C}$} \fi}
\newcommand{\Q}{\ifmmode\hspace{.2em}\vbar\hspace{-.31em}{\rm Q} \else \mbox{$\hspace{.2em}\vbar\hspace{-.31em}{\rm Q}$} \fi}
\newcommand{\Z}{\ifmmode{\rm Z}\hspace{-.28em}{\rm Z} \else ${\rm Z}\hspace{-.28em}{\rm Z}$ \fi}

\DeclareAcronym{AWGN}{
    short = AWGN,
    long = additive white Gaussian noise,
    list = Additive White Gaussian Noise,
    tag = abbrev
}

\DeclareAcronym{ADMM}{
    short = ADMM,
    long = alternating direction method of multipliers,
    list = Alternating Direction Method of Multipliers,
    tag = abbrev
}

\DeclareAcronym{MGMC}{
    short = MGMC,
    long = multi-group multi-casting,
    list = multi-group multi-casting,
    tag = abbrev
}

\DeclareAcronym{SGMC}{
    short = SGMC,
    long = single-group multi-casting,
    list = single-group multi-casting,
    tag = abbrev
}
\DeclareAcronym{AoA}{
    short = AoA,
    long = angle-of-arrival,
    list = Angle-of-Arrival,
    tag = abbrev
}

\DeclareAcronym{AoD}{
    short = AoD,
    long = angle-of-departure,
    list = Angle-of-Departure,
    tag = abbrev
}

\DeclareAcronym{KKT}{
    short = KKT,
    long = Karush-Kuhn-Tucker,
    list = Karush-Kuhn-Tucker,
    tag = abbrev
}

\DeclareAcronym{MMF}{
    short = MMF,
    long = max-min-fairness,
    list = max-min-fairness,
    tag = abbrev
}

\DeclareAcronym{WMMF}{
    short = WMMF,
    long = weighted max-min-fairness,
    list = max-min-fairness,
    tag = abbrev
}

\DeclareAcronym{BB}{
    short = BB,
    long = base band,
    list = Base Band,
    tag = abbrev
}

\DeclareAcronym{BC}{
    short = BC,
    long = broadcast channel,
    list = Broadcast Channel,
    tag = abbrev
}

\DeclareAcronym{BS}{
    short = BS,
    long = base station,
    list = Base Station,
    tag = abbrev
}

\DeclareAcronym{BR}{
    short = BR,
    long = best response,
    list = Best Response, 
    tag = abbrev
}

\DeclareAcronym{CB}{
    short = CB,
    long = coordinated beamforming,
    list = Coordinated Beamforming,
    tag = abbrev
}

\DeclareAcronym{CC}{
    short = CC,
    long = coded caching,
    list = Coded Caching,
    tag = abbrev
}

\DeclareAcronym{CE}{
    short = CE,
    long = channel estimation,
    list = Channel Estimation,
    tag = abbrev
}

\DeclareAcronym{CoMP}{
    short = CoMP,
    long = coordinated multi-point transmission,
    list = Coordinated Multi-Point Transmission,
    tag = abbrev
}

\DeclareAcronym{CRAN}{
    short = C-RAN,
    long = cloud radio access network,
    list = Cloud Radio Access Network,
    tag = abbrev
}

\DeclareAcronym{CSE}{
    short = CSE,
    long = channel specific estimation,
    list = Channel Specific Estimation,
    tag = abbrev
}

\DeclareAcronym{CSI}{
    short = CSI,
    long = channel state information,
    list = Channel State Information,
    tag = abbrev
}

\DeclareAcronym{CSIT}{
    short = CSIT,
    long = channel state information at the transmitter,
    list = Channel State Information at the Transmitter,
    tag = abbrev
}

\DeclareAcronym{CU}{
    short = CU,
    long = central unit,
    list = Central Unit,
    tag = abbrev
}

\DeclareAcronym{D2D}{
    short = D2D,
    long = device-to-device,
    list = Device-to-Device,
    tag = abbrev
}

\DeclareAcronym{DE-ADMM}{
    short = DE-ADMM,
    long = direct estimation with alternating direction method of multipliers,
    list = Direct Estimation with Alternating Direction Method of Multipliers,
    tag = abbrev
}

\DeclareAcronym{DE-BR}{
    short = DE-BR,
    long = direct estimation with best response,
    list = Direct Estimation with Best Response,
    tag = abbrev
}

\DeclareAcronym{DE-SG}{
    short = DE-SG,
    long = direct estimation with stochastic gradient,
    list = Direct Estimation with Stochastic Gradient,
    tag = abbrev
}

\DeclareAcronym{DFT}{
	short = DFT,
	long = discrete fourier transform,
	list = Discrete Fourier Transform,
	tag = abbrev
}

\DeclareAcronym{DoF}{
    short = DoF,
    long = degrees of freedom,
    list = Degrees of Freedom,
    tag = abbrev
}

\DeclareAcronym{DL}{
    short = DL,
    long = downlink,
    list = Downlink,
    tag = abbrev
}

\DeclareAcronym{GD}{
	short = GD, 
	long = gradient descent,
	list = Gradeitn Descent,
	tag = abbrev
}

\DeclareAcronym{IBC}{
    short = IBC,
    long = interfering broadcast channel,
    list = Interfering Broadcast Channel,
    tag = abbrev
}

\DeclareAcronym{i.i.d.}{
    short = i.i.d.,
    long = independent and identically distributed,
    list = Independent and Identically Distributed,
    tag = abbrev
}

\DeclareAcronym{JP}{
    short = JP,
    long = joint processing,
    list = Joint Processing,
    tag = abbrev
}

\DeclareAcronym{LOS}{
	short = LOS,
	long = line-of-sight,
	list = Line-of-Sight,
	tag = abbrev
}

\DeclareAcronym{LS}{
    short = LS,
    long = least squares,
    list = Least Squares,
    tag = abbrev
}

\DeclareAcronym{LTE}{
    short = LTE,
    long = Long Term Evolution,
    tag = abbrev
}

\DeclareAcronym{LTE-A}{
    short = LTE-A,
    long = Long Term Evolution Advanced,
    tag = abbrev
}

\DeclareAcronym{MIMO}{
    short = MIMO,
    long = multiple-input multiple-output,
    list = Multiple-Input Multiple-Output,
    tag = abbrev
}

\DeclareAcronym{MISO}{
    short = MISO,
    long = multiple-input single-output,
    list = Multiple-Input Single-Output,
    tag = abbrev
}

\DeclareAcronym{MAC}{
    short = MAC,
    long = multiple access channel,
    list = Multiple Access Channel,
    tag = abbrev
}

\DeclareAcronym{MSE}{
    short = MSE,
    long = mean-squared error,
    list = Mean-Squared Error,
    tag = abbrev
}

\DeclareAcronym{MMSE}{
    short = MMSE,
    long = minimum mean-squared error,
    list = Minimum Mean-Squared Error,
    tag = abbrev
}

\DeclareAcronym{mmWave}{
	short = mmWave,
	long = millimeter wave,
	list = Millimeter Wave,
	tag = abbrev
}

\DeclareAcronym{MU-MIMO}{
    short = MU-MIMO,
    long = multi-user \ac{MIMO},
    list = Multi-User \ac{MIMO},
    tag = abbrev
}

\DeclareAcronym{OTA}{
    short = OTA,
    long = over-the-air,
    list = Over-the-Air,
    tag = abbrev
}

\DeclareAcronym{PSD}{
    short = PSD,
    long = positive semidefinite,
    list = Positive Semidefinite,
    tag = abbrev
}

\DeclareAcronym{QoS}{
	short = QoS,
	long = quality of service,
	list = Quality of Service,
	tag = abbrev
}

\DeclareAcronym{RCP}{
	short = RCP,
	long = remote central processor,
	list = Remote Central Processor,
	tag = abbrev
}

\DeclareAcronym{RRH}{
    short = RRH,
    long = remote radio head,
    list = Remote Radio Head,
    tag = abbrev
}

\DeclareAcronym{RSSI}{
    short = RSSI,
    long = received signal strength indicator,
    list = Received Signal Strength Indicator,
    tag = abbrev
}

\DeclareAcronym{RX}{
	short = RX,
	long = receiver,
	list = Receiver,
	tag = abbrev
}

\DeclareAcronym{SCA}{
    short = SCA,
    long = successive-convex-approximation,
    list = Successive-Convex-Approximation,
    tag = abbrev
}

\DeclareAcronym{SG}{
    short = SG,
    long = stochastic gradient,
    list = Stochastic Gradient,
    tag = abbrev
}

\DeclareAcronym{SIC}{
    short = SIC,
    long = successive interference cancellation,
    list = Successive Interference Cancellation,
    tag = abbrev
}

\DeclareAcronym{SNR}{
    short = SNR,
    long = signal-to-noise-ratio,
    list = Signal-to-Noise Ratio,
    tag = abbrev
}

\DeclareAcronym{SDR}{
    short = SDR,
    long = semi-definite-relaxation,
    list = semi-definite-relaxation,
    tag = abbrev
}

\DeclareAcronym{SINR}{
    short = SINR,
    long = signal-to-interference-plus-noise ratio,
    list = Signal-to-Interference-plus-Noise Ratio,
    tag = abbrev
}

\DeclareAcronym{SOCP}{
	short = SOCP, 
	long = second order cone program,
	list = Second Order Cone Program,
	tag = abbrev
}

\DeclareAcronym{SSE}{
    short = SSE,
    long = stream specific estimation,
    list = Stream Specific Estimation,
    tag = abbrev
}

\DeclareAcronym{SVD}{
	short = SVD,
	long = singular value decomposition,
	list = Singular Value Decomposition,
	tag = abbrev
}

\DeclareAcronym{TDD}{
	short = TDD,
	long = time division duplex,
	list = Time Division Duplex,
	tag = abbrev
}

\DeclareAcronym{TX}{
	short = TX,
	long = transmitter,
	list = Transmitter,
	tag = abbrev
}

\DeclareAcronym{UE}{
    short = UE,
    long = user equipment,
    list = User Equipment,
    tag = abbrev
}

\DeclareAcronym{UL}{
    short = UL,
    long = uplink,
    list = Uplink,
    tag = abbrev
}

\DeclareAcronym{ULA}{
	short = ULA,
	long = uniform linear array,
	list = Uniform Linear Array,
	tag = abbrev
}

\DeclareAcronym{UPA}{
    short = UPA,
    long = uniform planar array,
    list = Uniform Planar Array,
    tag = abbrev
}

\DeclareAcronym{WMMSE}{
    short = WMMSE,
    long = weighted minimum mean-squared error,
    list = Weighted Minimum Mean-Squared Error,
    tag = abbrev
}

\DeclareAcronym{WMSEMin}{
    short = WMSEMin,
    long = weighted sum \ac{MSE} minimization,
    list = Weighted sum \ac{MSE} Minimization,
    tag = abbrev
}

\DeclareAcronym{WBAN}{
	short = WBAN,
	long = wireless body area network,
	list = Wireless Body Area Network,
	tag = abbrev
}

\DeclareAcronym{WSRMax}{
    short = WSRMax,
    long = weighted sum rate maximization,
    list = Weighted Sum Rate Maximization,
    tag = abbrev
}
\newtheorem{thm}{Theorem}

\newtheorem{exmp}{Example}
\theoremstyle{definition}
\newtheorem{rem}{Remark}

\newcommand{\CF}[0]{{\mathcal{F}}}

\newcommand{\CZ}[0]{{\mathcal{Z}}}

\newcommand{\Ba}[0]{{\mathbf{a}}}

\newcommand{\Bk}[0]{{\mathbf{k}}}

\newcommand{\Bw}[0]{{\mathbf{w}}}
\newcommand{\Bx}[0]{{\mathbf{x}}}

\newcommand{\BA}[0]{{\mathbf{A}}}

\newcommand{\BV}[0]{{\mathbf{V}}}

\newcommand{\Brk}[0]{{\Bar{k}}}

\newcommand{\Brt}[0]{{\Bar{t}}}

\newcommand{\Sfp}[0]{{\mathsf{p}}}

\newcommand{\FillGray}[3]{\filldraw[gray!50](#3-1+0.1,#1-#2+0.1) rectangle (#3-0.1,#1-#2+1-0.1)}
\newcommand{\FillBlack}[3]{\filldraw[black!70](#3-1+0.1,#1-#2+0.1) rectangle (#3-0.1,#1-#2+1-0.1)}
\newcommand{\FillHatch}[3]{\fill[pattern=crosshatch, pattern color=black!65](#3-1,#1-#2)rectangle(#3,#1-#2+1)}
\ExplSyntaxOn

\NewDocumentCommand \vect { s o m }
 {
  \IfBooleanTF {#1}
   { \vectaux*{#3} }
   { \IfValueTF {#2} { \vectaux[#2]{#3} } { \vectaux{#3} } }
 }
\DeclarePairedDelimiterX \vectaux [1] {\lbrack} {\rbrack}
 { \, \dbacc_vect:n { #1 } \, }
\cs_new_protected:Npn \dbacc_vect:n #1
 {
  \seq_set_split:Nnn \l_tmpa_seq { , } { #1 }
  \seq_use:Nn \l_tmpa_seq { \enspace }
 }
\ExplSyntaxOff

\newcommand{\subparagraph}{}
\usepackage{titlesec}
\titlespacing\section{3pt}{6pt plus 4pt minus 2pt}{6pt plus 2pt minus 2pt}
\titlespacing\subsection{3pt}{4pt plus 4pt minus 2pt}{4pt plus 2pt minus 2pt}
\titlespacing\subsubsection{3pt}{3pt plus 4pt minus 2pt}{0pt plus 2pt minus 3pt}

\title{Coded Caching and Spatial Multiplexing Gain Trade-off in Dynamic MISO Networks}

\begin{document}

\author{\IEEEauthorblockN{Milad Abolpour, MohammadJavad Salehi, and Antti T\"olli} \\
\IEEEauthorblockA{
    Centre for Wireless Communications, University of Oulu, 90570 Oulu, Finland \\
    \textrm{E-mail: \{firstname.lastname\}@oulu.fi}
    \vspace{-15pt}}
\thanks{This work is supported by the Academy of Finland under grants no. 318927 (6Genesis Flagship), 319059 (CCCWEE), and 343586 (CAMAIDE), and by the Finnish Research Impact Foundation (Vaikuttavuuss\"a\"ati\"o) under the project 3D-WIDE.}
}

\maketitle

\begin{abstract}
The global caching gain of multi-antenna coded caching techniques can be also mostly achieved in dynamic network setups, where the cache contents of users are dictated by a central server, and each user can freely join or leave the network at any moment.  In the dynamic setup, users are assigned to a limited set of caching profiles and the non-uniformness in the number of users assigned to each profile is compensated during the delivery phase by either adding phantom users for multicasting or serving a subset of users with unicast transmissions. In this paper, we perform a thorough analysis and provide closed-form representations of the achievable degrees of freedom (DoF) in such hybrid schemes, and assess 
the inherent trade-off between the global caching and spatial multiplexing gains caused by either adding phantom users or serving parts of the data through unicasting.
\end{abstract}

\begin{IEEEkeywords}
coded caching; dynamic networks; multi-antenna communications
\end{IEEEkeywords}

\section{Introduction}

With the increasing volume and variety of multimedia content, wireless networks are constantly being challenged for supporting higher data rates and lower latency~\cite{rajatheva2020whiteshort}. Especially, new wireless immersive viewing application, as one of the expected key 6G drivers, will push network capabilities beyond what's achievable by current state of the art~\cite{mahmoodi2021non,salehi2022enhancing}. The work of Maddah-Ali and Niesen proposed coded caching as a method of leveraging cache content across the network to promote efficient delivery of such multimedia content~\cite{maddah2014fundamental}. Coded caching boosts the achievable rate by a multiplicative factor proportional to the cumulative cache capacity in the entire network
through multicasting carefully designed codewords to diverse groups of users.
Given the significance of multi-antenna communications in the evolution of next-generation networks~\cite{rajatheva2020whiteshort}, the cache-aided \ac{MISO} configuration was later investigated in~\cite{shariatpanahi2016multi,shariatpanahi2018physical}, revealing that
the same coded caching benefit could be accomplished alongside the spatial multiplexing gain.

Many later research works in the literature addressed critical scaling and performance challenges raised by single- and multi-stream coded caching mechanisms. For example, optimized beamformer design for improving the finite \ac{SNR} performance of \ac{MISO} systems was discussed in~\cite{tolli2017multi,salehi2019subpacketization}, while the subpacketization issue (i.e., the number of smaller parts each file should be split into) was studied in~\cite{lampiris2018adding,salehi2020lowcomplexity}. Moreover, adapting coded caching to location-dependent file request scenario was considered in~\cite{mahmoodi2021non}, and increasing the caching gain with receiver-side multiple antennas was studied in~\cite{salehi2021MIMO}.

Another critical impediment to implementing coded caching techniques is their reliance on prior knowledge of the number of users for determining the cache contents at each user, making the implementation challenging in dynamic networks where the users can freely join and leave the network. Although fully decentralized schemes are able to address this issue partially~\cite{maddah2015decentralized}, their achievable gain approaches the one of classical centralized schemes only when the number of users or files tends to infinity. Another approach is to use hybrid centralized/decentralized approaches that 
aim to compensate for network dynamicity with a controlled \ac{DoF} loss~\cite{jin2019new, salehi2021low}. Especially, for MISO links, the idea is to assign users to a limited number of caching profiles, resulting in a shared-cache setup as studied in~\cite{parrinello2019fundamental}.

In this work, we specifically consider the hybrid scheme in~\cite{salehi2021low} that compensates for the non-uniformness in the number of users assigned with each profile by either adding virtual, \emph{phantom} users or excluding a subset of excess users from multicast transmission and instead serving them with multi-user unicast beamforming  in an orthogonal resource block. In this scheme, the delivery phase comprises two consecutive steps that deliver parts of data through coded caching techniques and unicasting, thus enabling a trade-off between the global caching and spatial multiplexing gains. 
While the work in~\cite{salehi2021low} considers the finite-SNR communication regime and uses numerical simulations to show the small performance gap with the uniform, centralized case, it lacks a proper theoretical analysis of the achievable DoF.
In this regard, here we perform a thorough analysis and provide closed-form representations of the achievable DoF of the scheme in~\cite{salehi2021low}, clarifying the real DoF loss with respect to the uniformly distributed users case, and the inherent trade-off caused by adding phantom users or serving parts of data with orthogonal unicast transmissions. 



Throughout the paper, we use boldface lower- and upper-case letters to denote vectors and matrices, respectively. Sets are shown with calligraphic letters. 
$\BA[i,j]$ is the element at row $i$ and column $j$ of matrix $\BA$, and $\Ba[i]$ is the $i$-th elements of vector $\Ba$.
$[K]$ represents the set $\{1,2,...,K \}$.

\section{Coded Caching for Dynamic MISO Setups}
\label{section:dynamic_cc_description}
\subsection{System Model}
\label{section:sysmodel_inside}
As discussed in~\cite{salehi2021low}, we consider a \ac{MISO} setup where a multi-antenna transmitter serves a set of cache-enabled users. The transmitter can support a maximum spatial multiplexing gain of $\alpha$, and has access to a library $\CF$ of some equal-sized files. The cache memory at each user is large enough to store a portion $0<\gamma<1$ of the entire library. We assume $\gamma = \frac{\Brt}{P}$, where $\Brt$ and $P$ are natural numbers and $\mathrm{gcd}\left(\Brt,P \right)=1$. The users are free to join and leave the network at any moment. Upon joining the network, the user $k$ is assigned with a cache profile $\Sfp(k) \in [P]$, and its cache contents are updated following a \emph{content placement} algorithm. A graphical representation of the system model is presented in Figure~\ref{fig:systeml}.

The dynamicity of the networks results in a varying number of users in the network during the time.
At given time intervals, the users present in the network reveal their requested files (from the library $\CF$) to the transmitter. The transmitter then builds and transmits a set of codewords following a \emph{content delivery} algorithm, such that after the transmissions are concluded, all the users can decode their requested contents. 
In this paper, we consider the delivery process in a specific time interval with $K$ total users in the network, and assume the same process is repeated in each interval. 
Following the general approach in the literature, we consider the total \ac{DoF}, defined as the average number of users served simultaneously during each interval, as the performance metric.

\begin{figure}[t]
    \centering
    \includegraphics[height=3cm,width=5cm]{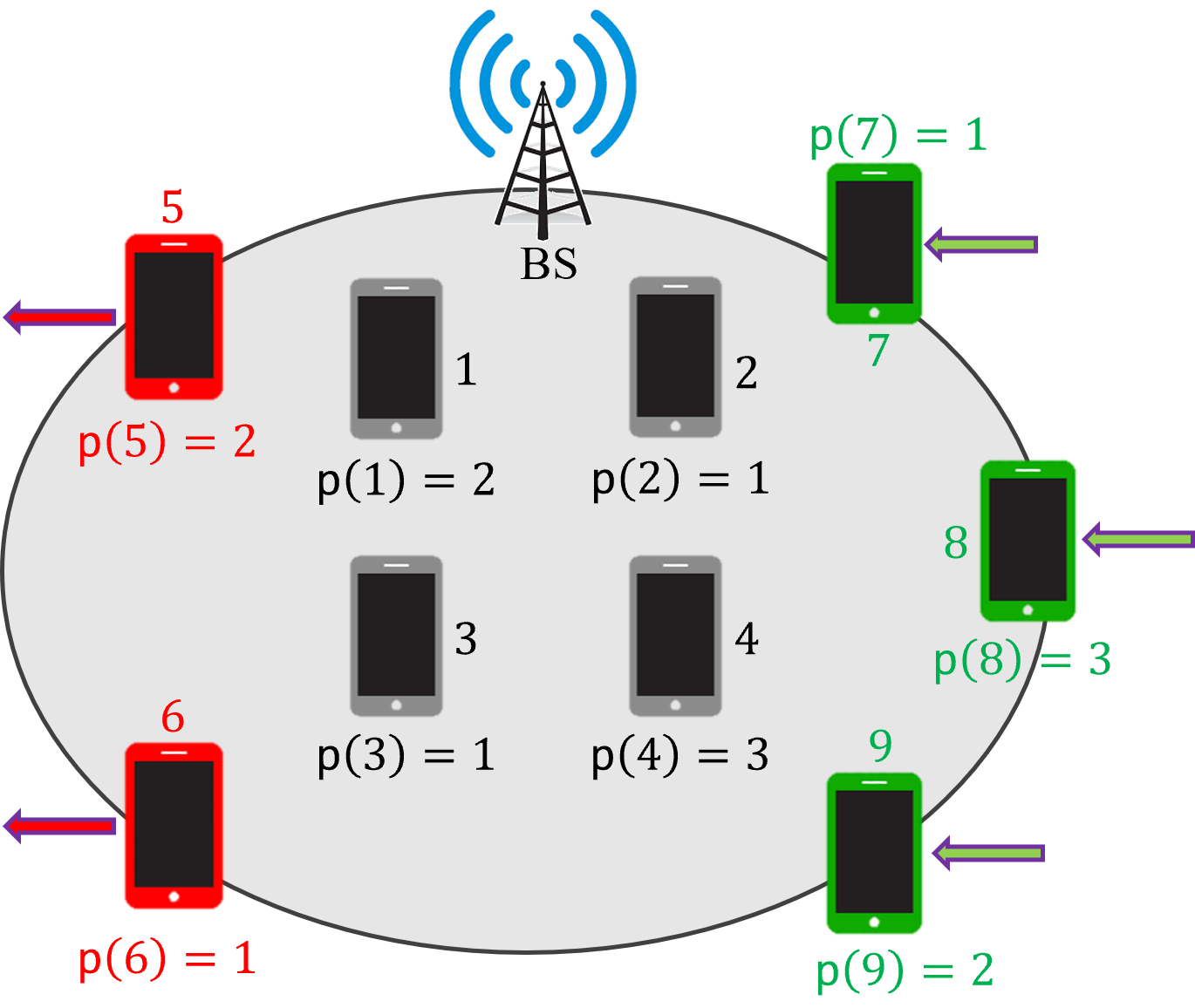}
    \caption{System model of the considered dynamic scheme at a sample time instant: Users $\left\lbrace 1,2,3,4 \right\rbrace$ are present and do not leave the network, users $\left\lbrace 7,8,9 \right\rbrace$ are joining, and users $\left\lbrace 5,6 \right\rbrace$ are leaving the network. Here, $\Sfp[k]$ is the profile that user $k$ is assigned to, e.g., user $1$ is assigned to profile $2$.}
    \label{fig:systeml}
\end{figure}

\subsection{Content Placement}
Upon joining the network, every user $k \in [K]$ is assigned with a caching profile $\Sfp(k)$, where $\Sfp(k) \in [P]$. Cache contents of user $k$ follow its assigned profile. Let us define a $P\times P$ binary \emph{placement matrix} $\BV$, for which: a) in the first row, the first $\overline{t}$ elements are one and the rest are zero, b) the row $p$, $1 < p \le P$, is a circular shift of its previous row to the right. Then, we split each file $W \in \CF$ into $P$ equal-sized \emph{packets} $W_p$, and if $\BV[p,\Sfp(k)] = 1$ for some row index $p$ and user index $k$, store $W_p$ in the cache memory of user $k$ for every file $W \in \CF$. Note that multiple users may be assigned with the same profile, and hence, have similar cache contents.

For the considered time interval, we assume the number of users assigned with profile $p$ is $\eta_p$, and $\sum_p \eta_p = K$. \footnote{During each time interval, depending on the number of users leaving or joining the network, $\eta_{p}$ and $K$ values vary.
} We call the parameter $\eta_p$ the \emph{length} of profile $p$. Cache placement is clarified more in Example~\ref{exmp:dynamic_cc_master}.


\subsection{Content Delivery}
Following~\cite{salehi2021low}, content delivery for the considered dynamic network setup comprises two consecutive steps. Here, we first briefly review each step, and then, clarify the whole process in Example~\ref{exmp:dynamic_cc_master}.

\subsubsection{Coded Caching (CC) data delivery}
In this step, the server first selects a \emph{unifying profile length} parameter $\hat{\eta}$, and then, assuming $\hat{\eta}$ users are assigned to each profile, builds and transmits a set of codewords using a novel coded caching scheme. We discuss the effect of the parameter $\hat{\eta}$ on the \ac{DoF} in the next section. In order to assume a uniform profile assignment, for every profile $p$,
\\$\bullet$ if $\hat{\eta} < \eta_{p}$, we randomly\footnote{In practice, we may get an improved performance by excluding users with poor channel conditions.} select $\eta_{p} - \hat{\eta}$ users assigned with profile $p$ and exclude them from the CC data delivery step. These users are then served during the subsequent step;
\\$\bullet$  if $\hat{\eta}\geq \eta_{p}$, we add $\hat{\eta}-\eta_{p}$ \textit{phantom} users, which are imaginary, non-existent users, and assign them to profile $p$. 

The codeword building process follows a similar approach to the RED scheme in~\cite{salehi2020lowcomplexity}. First, we consider a virtual MISO network with the spatial DoF $\Bar{\alpha}=\left\lceil \alpha / \hat{\eta} \right\rceil$ where each of the $P$  caching profiles is treated as a virtual user (thus resulting in the coded caching gain $\overline{t}=P\gamma$).
Then, we use the RED scheme in~\cite{salehi2020lowcomplexity} to create transmission vectors for the virtual network, resulting in $P$ transmission rounds each with $P-\Brt$ transmissions. We use $\bar{\Bx}_{j}^{r}$ to denote the virtual transmission vector $j$ at round $r$, and $\bar{\Bk}_{j}^{r}$ to represent the set of virtual users (profiles) served by $\bar{\Bx}_{j}^{r}$. Finally, we \emph{elevate} every virtual transmission vector $\bar{\Bx}_{j}^{r}$ and remove the effect of the phantom users to yield the transmission vectors for the original network.

Creating transmission vectors for the virtual network follows a cyclic delivery algorithm and is explained in detail in~\cite{salehi2020lowcomplexity}. We revisit this procedure in the next section while calculating the \ac{DoF}.
However, during the elevation process, we elevate each virtual transmission vector into $\hat{\eta}$ transmission vectors for the original network\footnote{In~\cite{salehi2021low}, it is suggested that if $\frac{\alpha}{\hat{\eta}}$ is an integer, each virtual vector can be elevated into only one vector for the original network. This suggestion reduces the required subpacketization by a factor of $\hat{\eta}$ but complicates the DoF analysis.} by replacing each virtual user with either $\hat{\eta}$ or $b$ users, depending on the value of $b=\alpha- \hat{\eta}\left\lfloor \alpha/\hat{\eta} \right\rfloor$, and removing or suppressing the inter-stream interference for these users with spatial multiplexing. For every user $k$ replacing a virtual user $\Brk$ we have $\Sfp(k) = \Brk$, and hence, we use the term caching profile interchangeably with the virtual user index while discussing the elevation process. As discussed in~\cite{salehi2021low}, during elevation, we also need to further split each packet into $\rho = \hat{\eta} \Brt + \alpha$ smaller \emph{subpackets}.

\subsubsection{Unicast (UC) data delivery}
In this step, the server transmits all missing data excluded from the CC data delivery phase (i.e., the data requested by users excluded when $\hat{\eta} < \eta_p$ for some profile $p$). At this step, only the local caching gain is available together with the spatial multiplexing gain, and there is no global caching gain as in the CC delivery step.

Let us assume the number of users being served in this step is $K_U$. We split each file requested by these users into the same number of subpackets as in the CC delivery step. As the cache ratio at each user is $\gamma$ and the total subpacketization is $P\rho$, each user needs to receive (at most) $P\rho(1-\gamma)$ subpackets to decode its requested file. To deliver the missing subpackets, we use a greedy algorithm in which we: 1) sort the users by their number of missing subpackets, 2) build a transmission vector that delivers one missing subpacket to each of the first $\alpha$ users (or all the remaining users, if  $K_{U} < \alpha$) using spatial multiplexing, and 3) repeats the whole procedure until all the subpackets are delivered.

\begin{figure}[t]
     \centering
     \begin{subfigure}[t]{0.4 \columnwidth}
         \centering
     \includegraphics[scale=0.3]{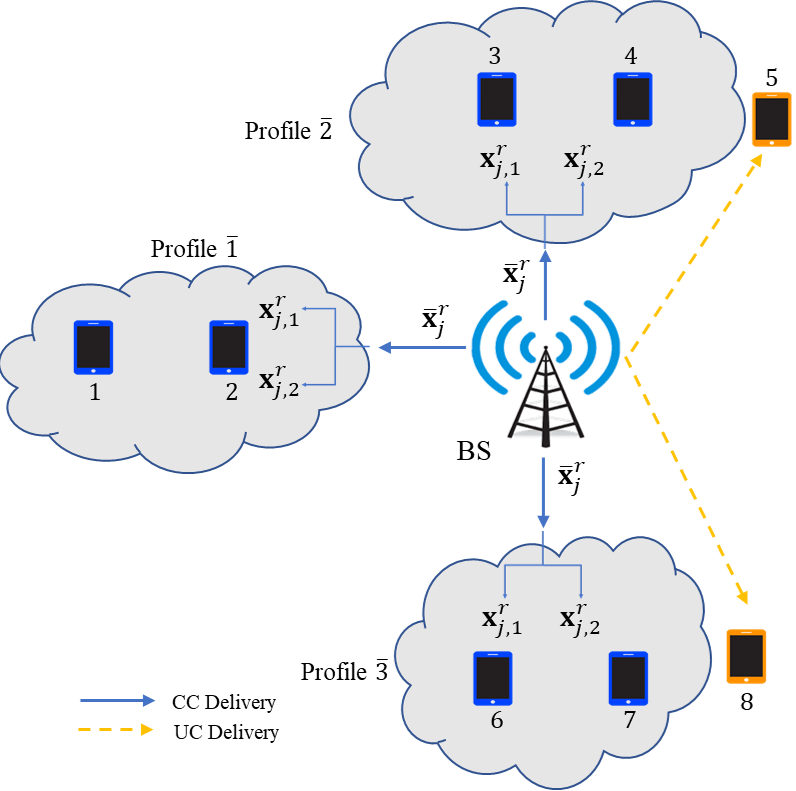}
 \caption{$\hat{\eta}=2$}
         \vspace{2em}
         \label{fig:eta2}
     \end{subfigure}
     \hfill
     \begin{subfigure}[t]{0.4 \columnwidth}
         \centering
 \includegraphics[scale=0.3]{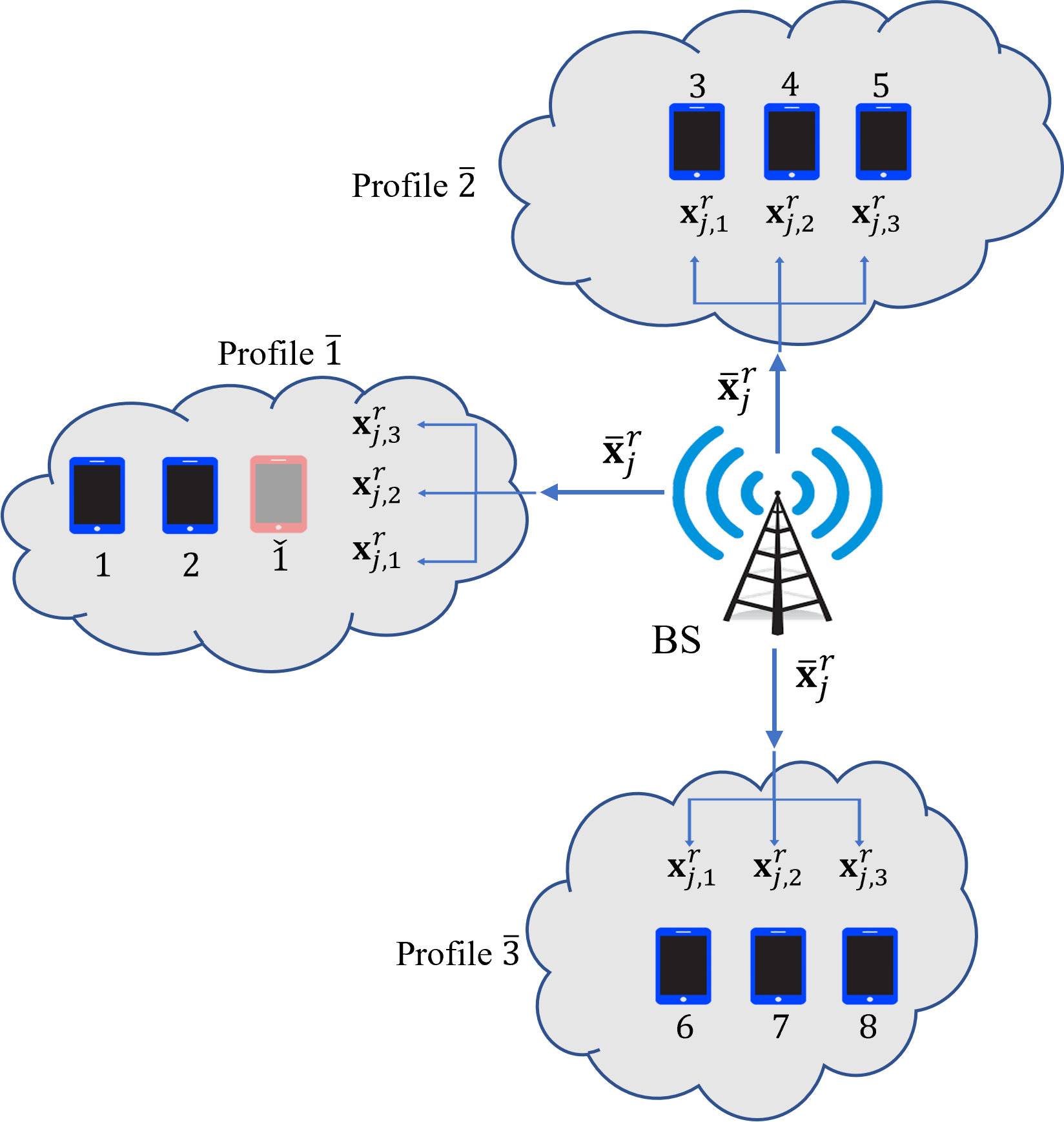}
         \caption{$\hat{\eta}=3$}
         \label{fig:eta3}
     \end{subfigure}
     \vspace{-2em}
     \caption{Content delivery phase for the network in Example~\ref{exmp:dynamic_cc_master}, with $\hat{\eta}=2$ and $\hat{\eta}=3$.
     Blue- and orange-colored icons represent users served in CC and UC delivery steps, respectively. Also, the phantom user is distinguished with the light red color. 
     }
     \label{fig:exmaple_dynamic_master}
\end{figure}

\begin{exmp}
\label{exmp:dynamic_cc_master}
Consider a  MISO network with $\alpha = 4$ and $\gamma = \frac{1}{3}$ (i.e., $P=3$ and $\Brt=1$). The placement matrix $\BV$ and the cache contents of the users assigned with the first profile (shown by $\CZ(1)$) are given as
\begin{equation*}
    \BV = 
    \renewcommand\arraystretch{0.6}\begin{bmatrix}
   1 & 0 & 0 \\
    0 & 1 & 0 \\
    0 & 0 & 1
    \end{bmatrix},
    \quad
    \CZ(1) = \{W_1 \; | \; W \in \CF \}.
\end{equation*}
Let us assume that in the considered time instant, eight users are in the network, and users $\{1,2\}$, $\{3,4,5\}$, and $\{6,7,8\}$ have been assigned with the first, second, and third profiles, respectively. For this network, we briefly review the delivery process for $\hat{\eta} = 2,3$. As $\bar{\alpha}=\left\lceil \alpha/  \hat{\eta} \right\rceil$, for both cases we have $\bar{\alpha}=2$. Then, the virtual network will have $P=3$ virtual users, coded caching gain $\Brt =1$, and spatial multiplexing gain $\bar{\alpha} = 2$. Using the RED scheme in~\cite{salehi2020lowcomplexity}, the packet index vector for the first transmission in the first round for this virtual network is $\bar{\Bk}_{1}^{1}=\vect{\bar{1},\bar{2},\bar{3}}$,
and the respective virtual transmission vector is
\begin{equation}
\label{eq: first_round_transmission vects exmp}
\begin{aligned}
&\bar{\Bx}_{1}^{1}=\bar{W}_{2}^{1}\left( \bar{1} \right)\bar{\Bw}_{\bar{3}}+\bar{W}_{1}^{1}\left( \bar{2} \right)\bar{\Bw}_{\bar{3}}+\bar{W}_{1}^{1}\left( \bar{3} \right)\bar{\Bw}_{\bar{2}},\\
%
\end{aligned}
\end{equation}
where $\bar{W}\left( \Brk \right)$ denotes the virtual file requested by virtual user $\Brk$, superscripts show the subpacket index, and $\Bar{\Bw}_{\Brk}$ is the optimized beamformer suppressing its associated data term at virtual user (profile) $\Brk$.

Now, if $\hat{\eta} = 2$, we don't need to add any phantom users but should exclude two real  users assigned with the second and third profiles. Let us assume users 5 and 8 are excluded. 
Then, the first transmission vector resulting from elevating $\bar{\Bx}_1^1$ is 
\begin{equation}
\label{eq:trans_eta2}
        \begin{aligned}
            \Bx_{1,1}^{1}&=W_{2}^{1}(1)\Bw_{6,7,2}+W_{2}^{1}(2)\Bw_{6,7,1}+W_{1}^{1}(3)\Bw_{6,7,4}\\
            &+W_{1}^{1}(4)\Bw_{6,7,3}+W_{1}^{1}(6)\Bw_{3,4,7}+W_{1}^{1}(7)\Bw_{3,4,6}.
        \end{aligned}
\end{equation}
The excluded users 5 and 8 are then served in the UC step, which requires $P \left(\hat{\eta}\bar{t}+\alpha \right) (1-\gamma) =12$ transmissions each delivering one subpacket in parallel to both users 5 and 8.

However, if $\hat{\eta} = 3$, we need to add one phantom user and assign it to the first profile, but no user is excluded from the CC delivery step (and hence, there is no UC step). Let us assume the phantom user $\check{1}$ is assigned to the first profile. Then, the first transmission vector resulting from the elevation of $\bar{\Bx}_1^1$ is
\begin{equation}
    \label{eq: sub_vectors b1 exmp}
        \begin{aligned}
            \Bx_{1,1}^{1}&=W_{2}^{1}(1)\Bw_{6,2,\check{1}}+W_{2}^{1}(2)\Bw_{6,1,\check{1}}+W_{2}^{1}(\check{1})\Bw_{6,1,2}\\
            &+W_{1}^{1}(3)\Bw_{6,4,5}+W_{1}^{1}(4)\Bw_{6,3,5}+W_{1}^{1}(5)\Bw_{6,3,4}\\
            &+W_{1}^{1}(6)\Bw_{3,4,5},\\
        \end{aligned}
\end{equation}
which, after removing the effect of the phantom user $\check{1}$ is equivalent to
\begin{equation}
\label{eq:trans_eta3_nophantom}
        \begin{aligned}
            \Bx_{1,1}^{1}&=W_{2}^{1}(1)\Bw_{6,2}+W_{2}^{1}(2)\Bw_{6,1}
            +W_{1}^{1}(3)\Bw_{6,4,5}\\
            &+W_{1}^{1}(4)\Bw_{6,3,5}+W_{1}^{1}(5)\Bw_{6,3,4}+W_{1}^{1}(6)\Bw_{3,4,5}.
        \end{aligned}
\end{equation}

Let us consider the decoding process at user one, after the transmission of $\Bx_{1,1}^1$ in~\eqref{eq:trans_eta3_nophantom}. Following the cache placement process, interference terms $W_{1}^{1}(3)$, $W_{1}^{1}(4)$, $W_{1}^{1}(5)$, $W_{1}^{1}(6)$ are all available in the cache memory of user one and could be removed from the received signal. On the other hand, $W_{2}^{1}(2)$ is also suppressed at user one with beamforming, and hence, this user can decode $W_{2}^{1}(1)$ interference-free. 

The transmission process for both cases of $\hat{\eta} = 2,3$ is depicted in Figure~\ref{fig:exmaple_dynamic_master}. Comparing equations~\eqref{eq:trans_eta2}-\eqref{eq:trans_eta3_nophantom}, we can see that without phantom users, the DoF value of the CC step is fixed and limited, but we need to deliver part of the data in the UC step (which hurts the total DoF). On the other hand, by increasing $\hat{\eta}$, a fewer number of users are served in the UC step, but the DoF in the CC step varies in each transmission.

\end{exmp}

\section{DoF Analysis}
\label{section:DoFAnalysis}
As mentioned in Section~\ref{section:sysmodel_inside}, we define DoF as the average number of users served simultaneously  during the whole content delivery phase. However, as every target user is receiving exactly one subpacket during each transmission, we may also equivalently calculate the average number of transmitted subpackets instead. In this regard, defining $J_{M}$ and $J_{U}$ to be the total number of transmitted subpackets, and $T_{M}$ and $T_{U}$ to denote the total number of transmissions at CC and UC delivery steps, respectively, the DoF is calculated as
\begin{equation}
\label{eq: DoF def}
    \text{DoF}=\frac{J_{M}+J_{U}}{T_{M}+T_{U}}.
\end{equation}
Let us also define $K_{M}$ and $K_{U}$ to be the number of users being served in CC and UC data delivery steps, respectively. In other words, given the $\hat{\eta}$ value, we have $K_M = \sum_p \min (\hat{\eta}, \eta_p)$ and $K_U = \sum_p \max (\eta_p-\hat{\eta},0)$.

In order to compute the DoF, we need to know the number of times each virtual user (or equivalently, each profile) is served during the CC data delivery step. Lemma~\ref{lemma: count users} provides us with that.
\begin{lemma}
\label{lemma: count users}
Each virtual user appears $\left( P-\bar{t} \right)$ times as the $m$-th element, $m\in \left[ \bar{t}+\bar{\alpha} \right]$, in all user index sets $\bar{\Bk}_{j}^{r}$, $r \in [P]$ and $j \in [P-\bar{t}]$. 
\end{lemma}
\begin{proof}
From~\cite{salehi2020lowcomplexity}, we know that for every round $r \in \left[ P \right]$ and transmission index $j \in \left[ P-\bar{t} \right]$, $\bar{\Bk}_{j}^{r}$ is calculated as

\begin{equation*}
\label{eq: k bar def}
    \begin{aligned}
    \bar{\Bk}_{j}^{r} = &\Big[  \big[ r:\left( r+ \bar{t}-1\right) \bmod{P} \big];  \\
   &  \Big( \big( \left( \left[ r: r+\bar{\alpha}  \right] +j-1 \right) \bmod{\left( P-\bar{t} \right)}\big) +\bar{t} \Big) \bmod{P} \Big].
    \end{aligned}
\end{equation*}
Let us consider a fixed round $r=p$. Clearly, the first element of $\bar{\Bk}_{j}^{p}$ is $p$, and as $\bar{\Bk}_{j}^{p+1}=\left( \bar{\Bk}_{j}^{p}+1 \right) \bmod{P}$, the second element of $\bar{\Bk}_{j}^{p+1}$ is also $p$. Similarly, as for a fixed $j$ each $\bar{\Bk}_{j}^{r}$ is a circular shift of $\bar{\Bk}_{j}^{p}$, the first element of $\bar{\Bk}_{j}^{p}$ is the $(\left( r-p \right)\bmod{P}+1)$-th element of $\bar{\Bk}_{j}^{r}$. As a result, for a fixed $j$, the first element of $\bar{\Bk}_{j}^{p}$ equally places as the $r$-th element of $\bar{\Bk}_{j}^{r}$ during all the $P$ rounds of transmissions. On the other hand, since $j$ varies from one to $\left( P-\bar{t} \right)$, the number of times that the first element of $\bar{\Bk}_{j}^{p}$ (and equivalently each virtual user)  appears as the  $m$-th element of  all user index sets $\bar{\Bk}_{j}^{r}$ is equal to $\left( P-\bar{t} \right)$.
\end{proof}
%

\begin{thm}
\label{Theorem: DoF}
The DoF of the proposed scheme in a network with spatial multiplexing gain $\alpha$, cache ratio $\gamma$, and the unifying profile length parameter $\hat{\eta}$ is equal to 
   %
   %
\begin{equation}
\label{eq: DoF bn0}
\mathrm{DoF}=
    \begin{cases}
   K_{M}\gamma+\frac{K_{M}\alpha}{P\hat{\eta}} & K_{U}=0 \\
       \frac{K_{M} \left( P-\bar{t} \right)\left( \hat{\eta}\bar{t}+\alpha \right) +K_{U}\left( P-P\gamma \right)\left( \hat{\eta}\bar{t}+\alpha \right)}{P\left( P -\bar{t} \right)\hat{\eta}+ \left( P-P\gamma \right)\left( \hat{\eta}\bar{t}+\alpha \right)} & K_{U} < \alpha \\
    \frac{K_{M} \left( P-\bar{t} \right)\left( \hat{\eta}\bar{t}+\alpha \right) +K_{U}\left( P-P\gamma \right)\left( \hat{\eta}\bar{t}+\alpha \right)}{P\left( P -\bar{t} \right)\hat{\eta}+ \left\lceil\frac{ K_{U} \left( P-P\gamma \right)\left( \hat{\eta}\bar{t}+\alpha \right)} {\alpha}\right\rceil} &  K_{U} \geq \alpha
    \end{cases}
\end{equation}
\end{thm}
\begin{proof}
We explicitly derive $J_{M}$, $J_{U}$, $T_{M}$ and $T_{U}$,
and then, use the DoF definition in~\eqref{eq: DoF def}. In order to calculate $J_{M}$ and $T_{M}$ corresponding to the CC data delivery step, we recall that: 1) each  user in the virtual network corresponds to a caching profile in the original network; 
2) from Section~\ref{section:dynamic_cc_description}, every virtual transmission vector $\bar{\Bx}_{j}^{r}$, $j \in [P-\bar{t}], r \in [P]$,  is elevated into $\hat{\eta}$ transmission vectors $\Bx_{j,\delta}^{r}$, $\delta \in [\hat{\eta}]$, for the original network; 3) following the discussions in~\cite{salehi2021low}, the users connected to the first $\left( \bar{t}+\bar{\alpha}-1 \right)$ elements of $\bar{\Bk}_{j}^{r}$ are served with $\hat{\eta}$ transmission vectors $\Bx_{j,\delta}^{r}$, but the users assigned to the last element of $\bar{\Bk}_{j}^{r}$ are served only with $b$ transmission vectors $\Bx_{j,\delta}^{r}$ where $b$ is the remainder of the division of $\alpha$ by $\hat{\eta}$; 
4) using Lemma~\ref{lemma: count users}, every virtual user (i.e., profile index) appears $\left( P-\bar{t} \right)$ times as the $m$-th element of user index sets $\bar{\Bk}_{j}^{r}$, where $m \in [\Brt + \Bar{\alpha}]$. 
Now, since $K_{M}$ users participate in the CC delivery step, 
the total number of transmitted subpackets at this step is obtained as
\begin{equation}
\label{eq: JM bn0 def}
    J_{M}=K_{M} \left( P-\bar{t} \right)\big( \hat{\eta}\left( \bar{t}+\bar{\alpha}-1 \right)+b \big),
\end{equation}
which, after substituting $\bar{\alpha}=\left\lceil \alpha/\hat{\eta} \right\rceil$ and $b=\alpha- \hat{\eta}\left\lfloor \alpha/\hat{\eta} \right\rfloor$, is reduced to
\begin{equation}
\label{eq: JM bn0}
    J_{M}=K_{M} \left( P-\bar{t} \right) \left( \hat{\eta}\bar{t}+\alpha \right).
\end{equation}
In addition, as stated earlier, we have a total number of $P\left( P-\bar{t} \right)$ virtual transmission vectors, each elevated into $\hat{\eta}$ vectors for the original network. Hence, the total number of transmission in the CC delivery step is equal to
\begin{equation}
\label{eq: TM bn0}
    T_{M}=P\left( P-\bar{t} \right)\hat{\eta}.
\end{equation}


Now, for the UC delivery step, considering that each of the $K_U$ users requires $P\rho(1-\gamma)$ subpackets, we  have
\begin{equation}
 \label{eq: JU b}
 J_{U}=K_{U}\left(P-P\gamma\right)\rho.
 \end{equation}
Moreover, as the spatial multiplexing gain is $\alpha$, the server is able to serve $\min \left( K_{U}, \alpha \right)$ users at any given time instant, and hence, the total number of required transmissions at the UC delivery step would be
\begin{equation}
\label{eq; TU b}
    T_{U}=\left\lceil \frac{K_{U}\left( P-P\gamma \right)\rho}{\min \left( K_{U}, \alpha \right)} \right\rceil,
\end{equation}
%

Finally, the DoF value in Theorem~\ref{Theorem: DoF} can be calculated by substituting \eqref{eq: JM bn0}-\eqref{eq; TU b} into \eqref{eq: DoF def}.
\end{proof}

\begin{rem}
In case of full multicasting (i.e., $\hat{\eta} = \max \eta_p$, $K_{U}=0$), we have $K_M = K = P \eta_{\rm avg}$, and the DoF is 
\begin{equation}
    \mathrm{DoF} = K \gamma+\frac{K \alpha}{P\hat{\eta}} = K \gamma + \alpha \frac{\eta_{\rm avg}}{\hat{\eta}}.
\end{equation}
Comparing this DoF value with the two extreme cases of: 1) relying only on the spatial multiplexing gain, for which the DoF of $\alpha$ is achievable, and 2) uniform distribution of $\eta_p$ values and using coded caching techniques, for which the DoF of $K\gamma + \alpha$ is achievable, the DoF of the proposed scheme lies between the two extremes and enables an added caching gain of $K \gamma$ by incurring a loss in the spatial multiplexing gain by a multiplicative factor of $(1-\frac{\eta_{\rm avg}}{\hat{\eta}})$. The loss in the DoF is caused by phantom users, and is increased as the distribution of $\eta_p$ parameters becomes more non-uniform (i.e., if we need to add more phantom users).

\end{rem}

\section{Simulation Results}
\label{section:SIMResults}
Here, we investigate the achievable DoF  during the content delivery phase in a dynamic network setup with the spatial multiplexing gain $\alpha=50$ and the cache ratio $\gamma=\frac{1}{4}$ (i.e., $P=4$ profiles and $\bar{t}=1$). We assume $K=100$ users are present in the network, and the profile length values $\eta_p$ vary from zero to $\max_{p\in [4]} \eta_{p}$.

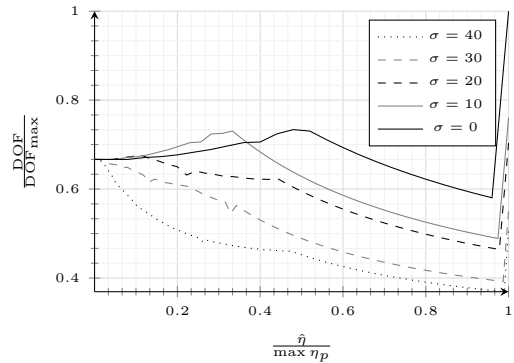
\begin{figure}[htb]
  \centering
        \begin{tikzpicture}

    \begin{axis}
    [
    width=0.8\columnwidth,
    height=0.6\columnwidth,
    axis lines = center,
    xlabel near ticks,
    xlabel = \smaller {\tiny $\frac{\hat{\eta}}{\max \eta_p}$},
    ylabel = \smaller {\tiny $\frac{\mathrm{DOF}}{\mathrm{DOF}_{\max}}$},
    ylabel near ticks,
    legend pos= north east,
    ticklabel style={font=\tiny},
    grid=both,
    major grid style={line width=.2pt,draw=gray!30},
    grid style={line width=.1pt, draw=gray!10},
    minor tick num=5,
    ]
    
    
    \addplot
    [dotted, black]
    table[y=DOF40,x=ETA40]{Figs/3D_Data.tex};
    \addlegendentry{\tiny $\sigma=40$}
    
    \addplot
    [dashed,gray]
    table[y=DOF30,x=ETA30]{Figs/3D_Data.tex};
    \addlegendentry{\tiny $\sigma=30$}
    
    \addplot
    [dashed,black]
    table[y=DOF20,x=ETA20]{Figs/3D_Data.tex};
    \addlegendentry{\tiny $\sigma=20$}
    
    \addplot
    [gray]
    table[y=DOF10,x=ETA10]{Figs/3D_Data.tex};
    \addlegendentry{\tiny $\sigma=10$}
    
    \addplot
    [black]
    table[y=DOF0,x=ETA0]{Figs/3D_Data.tex};
    \addlegendentry{\tiny $\sigma=0$}
    
    \end{axis}

    \end{tikzpicture}
    \caption{$\mathrm{DoF}/\mathrm{DoF}_{\max}$ versus $\hat{\eta}/\max \eta_{p}$ for different values of $\sigma$.}
    \label{fig: Dof 3D}
\end{figure}

Figure~\ref{fig: Dof 3D} illustrates how the (normalized) DoF varies with respect to $\hat{\eta}$, as the distribution of $\eta_p$ values becomes more non-uniform (i.e., the standard deviation of the $\eta_p$ distribution, shown by $\sigma$, is increased).
Here,
    $\sigma^2={\frac{1}{4}\sum_{p=1}^{4}\left( \eta_{p} -\eta_{\rm avg} \right)^{2}}$,
$\eta_{\rm avg} = 25$, and $\mathrm{DoF}_{\max}$ is the largest achievable DoF (found by exhaustive search).
As observed, when  $\sigma$ is small, the maximum DoF can be easily achieved by setting $\hat{\eta} = \max_p \eta_p$.
In other words, 
when the users are distributed approximately uniformly among the profiles, a better DoF is achieved by involving more users in the CC delivery step. This is due to the fact that for near-uniform distributions of $\eta_p$ values, setting $\hat{\eta} = \max_p \eta_p$ does not require adding many phantom users, while selecting $\hat{\eta}$ less than $\max \eta_{p}$ would increase the number of excluded users, thus decreasing the coded caching gain and hurting the DoF. 
%
However, for large values of $\sigma$ (e.g., $\sigma = 30$),  the maximum DoF can be achieved at $\hat{\eta}/\max\eta_{p} \approx 0$. which demonstrates that it is better to serve all users in the UC step to achieve a higher DoF. 
This is due to the fact that for very non-uniform distributions of $\eta_p$ values, selecting $\hat{\eta} > 0$ requires adding many phantom users and this hurts the DoF as their effects should be removed before the real transmissions.

It can also be observed from Figure~\ref{fig: Dof 3D} that the DoF is not a monotonic function of $\hat{\eta}$. This demonstrates that there exists a trade-off between the CC and UC delivery steps, which can be balanced by carefully choosing the $\hat{\eta}$ value.

\begin{figure}[htb]
    \centering
    \begin{tikzpicture}

\begin{axis}[%
width=0.8\columnwidth,
height=0.6\columnwidth,
axis lines = center,
xmin=0,
xlabel near ticks,
xlabel={\tiny Standard Deviation ($\sigma$)},
ymin=0.63,
ymax=1.02,
ylabel={\tiny ${\frac{\mathrm{DoF}_{M}}{\mathrm{DoF}_{{\mathrm{opt}}}}}$},
ylabel near ticks,
    grid=both,
    major grid style={line width=.2pt,draw=gray!30},
    grid style={line width=.1pt, draw=gray!10},
    minor tick num=5,
    legend pos = north east,
ticklabel style={font=\tiny},
]
\addplot[only marks, mark=x, mark size=1.1pt, black!50] table[y=DynamicY,x=DynamicXSD]{Figs/DoFV_Data.tex};
\addlegendentry{\tiny The proposed scheme}

\addplot [gray,line width=1.0pt]
table[y=UCY,x=UCXSD]{Figs/DoFV_Data_2.tex};
\addlegendentry{\tiny Unicast-only}

\addplot [black, dashed, line width=1.0pt]
table[y=MCY,x=MCXSD]{Figs/DoFV_Data_2.tex};
\addlegendentry{\tiny Multicast-only}

\end{axis}
\end{tikzpicture}%
    \caption{The ratio  $\mathrm{DoF}_{M}/\mathrm{DoF}_{\rm opt}$ versus $\sigma$}
    \label{fig:DoFV}
\end{figure}
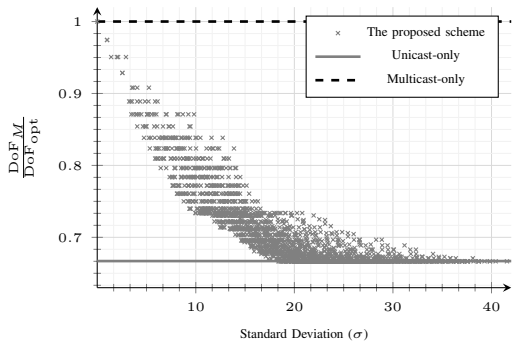

In Figure~\ref{fig:DoFV}, for every feasible distribution of $\eta_p$ values, we have compared the best achievable DoF (found by line search over the $\hat{\eta}$ value and shown by $\mathrm{DoF}_M$) with the DoF of the optimal, uniform distribution (denoted by \emph{multicast only} and shown by $\mathrm{DoF}_{\rm opt}$). For a better comparison, we have also included the achievable DoF when all the requested data is served using the UC delivery step only (i.e., without any coded caching gain).

As observed, the ratio $\mathrm{DoF}_{M}/\mathrm{DoF}_{\rm opt}$ decreases as $\sigma$ is increased. This shows that increasing non-uniformness in the distributions of $\eta_p$ values generally degrades the maximum achievable DoF with the proposed scheme. This was expected as both tools for compensating for the non-uniformness in $\eta_p$ distribution (i.e., adding phantom users and excluding users from the CC step) hurt the DoF. However, for moderate non-uniformness in $\eta_p$ values (e.g., $\sigma \le 10$), the proposed scheme enables a noticeable performance improvement (10-50 percent) over the full unicast method. The emphasizes the importance of employing various methods at the transmitter to make sure the distribution of $\eta_p$ values is as uniform as possible, while switching to simple unicast methods when the non-uniformness increases.

%
%
\section{Conclusion and Future Work}
We analyzed the degrees-of-freedom (DoF) of a hybrid decentralized/centralized coded caching scheme for dynamic setups where the users could freely join or leave the network at any moment. It was illustrated that increasing the non-uniformness  in users' distributions decreases the achievable DoF. Accordingly,  in order to compensate for this non-uniformness, we either added some phantom users  or served a subset of users with unicast transmissions which raised an inherent trade-off  between the global caching and spatial multiplexing gains. In other words, for each users' distribution, we should find the best unifying profile length to maximize the DoF of serving users with either multicast or unicast transmissions. Managing the trade-off between global caching  and spatial multiplexing gain to find the optimal profile length is a subject of our future work. 
\bibliographystyle{IEEEtran}
\bibliography{references,ref_whitepapershort}

\end{document}